\documentclass[sigplan,screen,nonacm]{acmart}

\usepackage{hyperref}
\usepackage{graphicx}
\usepackage{textcomp}
\usepackage{xcolor}
\usepackage{algorithm}
\usepackage{algorithmic}

\def\BibTeX{{\rm B\kern-.05em{\sc i\kern-.025em b}\kern-.08em
    T\kern-.1667em\lower.7ex\hbox{E}\kern-.125emX}}

\begin{document}

\title{Vectorization of Verilog Designs and its Effects on Verification and Synthesis}

\author[M. Guimar\~{a}es]{Maria Fernanda Oliveira Guimar\~{a}es}
\orcid{}
\affiliation{%
  \institution{UFMG}
  \city{Belo Horizonte}
  \country{Brazil}
}
\email{maria.guimaraes@dcc.ufmg.br}

\author[U. Rosa]{Ulisses Rosa}
\orcid{}
\affiliation{%
  \institution{UFMG}
  \city{Belo Horizonte}
  \country{Brazil}
}
\email{ulissesrosa@dcc.ufmg.br}

\author[I. Trudel]{Ian Trudel}
\orcid{}
\affiliation{%
  \institution{Independent Researcher}
  \city{Quebec}
  \country{Canada}
}
\email{ian.trudel@researchcenter.io}

\author[J. Vieira]{Jo\~{a}o Victor Amorim Vieira}
\orcid{}
\affiliation{%
  \institution{Cadence}
  \city{Belo Horizonte}
  \country{Brazil}
}
\email{jvamorim@cadence.com}

\author[A. Mafra]{Augusto Amaral Mafra}
\orcid{}
\affiliation{%
  \institution{Cadence}
  \city{Belo Horizonte}
  \country{Brazil}
}
\email{augusto@cadence.com}

\author[M. Crepalde]{Mirlaine Crepalde}
\orcid{}
\affiliation{%
  \institution{Cadence}
  \city{Belo Horizonte}
  \country{Brazil}
}
\email{mirlaine@cadence.com}

\author[F. Pereira]{Fernando Magno Quint\~{a}o Pereira}
\orcid{}
\affiliation{%
  \institution{UFMG}
  \city{Belo Horizonte}
  \country{Brazil}
}
\email{fernando@dcc.ufmg.br}

\begin{abstract}
Vectorization is a compiler optimization that replaces multiple operations on scalar values with a single operation on vector values.
Although common in traditional compilers such as rustc, clang, and gcc, vectorization is not common in the Verilog ecosystem.
This happens because, even though Verilog supports vector notation, the language provides no semantic guarantee that a vectorized signal behaves as a word-level entity: synthesis tools still resolve multiple individual assignments and a single vector assignment into the same set of parallel wire connections.
However, vectorization brings important benefits in other domains.
In particular, it reduces symbolic complexity even when the underlying hardware remains unchanged.
Formal verification tools such as Cadence\textsuperscript{\textregistered} Jasper\textsuperscript{\textregistered} operates at the symbolic level: they reason about Boolean functions, state transitions, and equivalence classes, rather than about individual wires or gates.
When these tools can treat a bus as a single symbolic entity, they scale more efficiently.
This paper supports this observation by introducing a Verilog vectorizer.
The vectorizer, built on top of the CIRCT compilation infrastructure, recognizes several vectorization patterns, including inverted assignments, assignments involving complex expressions, and inter-module assignments.
It has been experimented with some Electronic design automation (EDA) tools, and for Jasper tool, it improves elaboration time by 28.12\% and reduces memory consumption by 51.30\% on 1,157 designs from the ChiBench collection.
\end{abstract}



\maketitle

\section{Introduction}
\label{sec_intro}

Vectorization is a classic compiler optimization that replaces scalar operations with equivalent operations over wider data types~\cite{Allen88}.
It is used in software compilers such as \texttt{gcc} and \texttt{clang} at the \texttt{-O2} and \texttt{-O3} optimization levels to improve runtime~\cite{Alves15}.
Whenever applicable, vectorization is very effective.
In the words of Maleki {\it et al.}~\cite{Maleki11}, ``{\it Vectorization is one of the compiler transformations that can have significant impact on performance.}''

Nevertheless, in spite of its importance in classic compilers, vectorization is absent in compilers and tools that process Verilog designs.
At the hardware level, vectors are ultimately realized as independent wires and bit-level logic, so synthesis tools lower vector signals into scalars early in the flow.
This practice has created a common assumption in the EDA community that recovering vector structure offers little benefit. Consequently, mainstream Verilog processing frameworks such as CIRCT (Circuit Intermediate Representation Compilers and Tools) and Yosys do not provide analyses or transformations that group related scalar signals back into structured vectors.

\paragraph{Vectorization of Hardware-Description Languages.}
This paper demonstrates that vectorization can be beneficial even when applied in the context of hardware-description languages.
Specifically, it provides evidence in support of the following thesis:
\begin{small}
\begin{quotation}
{\it ``Although the vectorization of Verilog designs does not change the hardware they describe, it reduces their symbolic complexity, enabling faster and more scalable analysis and verification.''}
\end{quotation}
\end{small}
To substantiate this claim, Section~\ref{sec_ovf} introduces three transformations tailored to Verilog designs, while Section~\ref{sec_impl} describes how these transformations are realized through static compiler analyses and source-to-source rewriting.

\paragraph{Spatial vs Temporal Vectorization.}
The form of vectorization proposed in this paper is different from classic vectorization.
In both settings, vectorization consists of identifying recurring computation patterns and collapsing them into a single, more compact representation.
Classic compiler vectorization, however, focuses on \emph{temporal recurrences}, that is, repeated executions of the same operations across iterations of a loop.
This view originates in the seminal work of Allen {\it et al.}~\cite{Allen88} and remains central to modern vectorizers~\cite{Alladi26}.

In contrast, the vectorization presented in this paper targets \emph{spatial recurrences} that arise in RTL designs.
As Section~\ref{sec_ovf} illustrates, Verilog descriptions often express word-level behavior through replicated bit-level assignments that differ only in constant indices.
Our approach detects groups of assignments that are structurally equivalent but operate on distinct signals, and rewrites them as single vector-level expressions.
Thus, while traditional vectorization is loop-centric, our technique identifies structurally replicated logic in the absence of explicit iteration constructs.
From this perspective, the proposed vectorizer is conceptually closer to techniques that recognize and outline recurrent code patterns~\cite{Ablego07,Campos23} than to classical loop-based vectorization.

\paragraph{Summary of Results.}
Section~\ref{sec_eval} describes a vectorizer implemented within the CIRCT infrastructure.
We have evaluated this implementation over a collection of 1,157 designs from the ChiBench~\cite{Sumitani24} collection. Section~\ref{sub_rq3} shows that vectorization improves Jasper's RTL (Register-Transfer Level) elaboration time by over 28\% and reduces its memory consumption by more than 50\%. Section~\ref{sub_rq4} shows that vectorization also reduces Cadence Genus's elaboration time by more than 5\% on average, with much larger improvements on individual benchmarks.
These results have motivated the CIRCT community to accept the proposed vectorizer as a standard optimization, which has been available in their open-source distribution since early 2026.

\section{Overview}
\label{sec_ovf}

Vectorization of Verilog designs is the process of detecting groups of bit-level operations that collectively implement word-level expressions and replacing them with equivalent vector operations.
The implementation discussed in this paper supports any combination of three classes of transformations, defined as follows:
\begin{description}
\item [General bit permutations:]
Arbitrary reorderings of signals are compacted into a single vector assignment.

\item [Intermodule vectorization:]
Structurally similar module instances are selectively inlined to expose vectorization opportunities across module boundaries.

\item [Complex patterns:]
Logic and arithmetic operations defined over individual bits are combined into a single vector-level expression.
\end{description}
The remainder of this section illustrates each of these patterns through a set of examples, starting with Example~\ref{ex_general_pattern}.

\begin{example}
\label{ex_general_pattern}
Figure~\ref{fig_general_pattern} shows the effect of vectorizing a sequence of bit assignments.
The vectorization of general assignments targets designs in which the bits of an output vector form an arbitrary reordering of the bits of a single input vector.
Although such designs are often written as multiple independent bit-level assignments, they collectively implement a one-to-one permutation.
In this example, if we let $\pi : \{0, 1, 2, 3\} \rightarrow \{ 3, 1, 2, 0 \}$ be a permutation, then the vectorizer will implement the assignment $\mathtt{out}[i] = \mathtt{in}[\pi(i)]$.
\end{example}

\begin{figure}[ht]
\includegraphics[width=1.0\columnwidth]{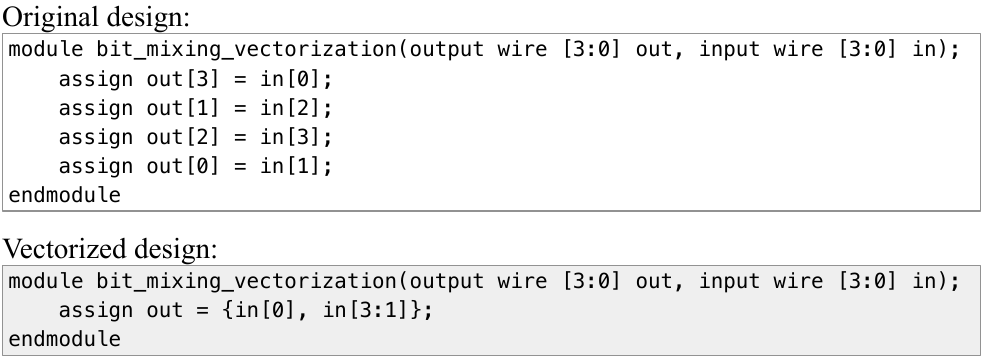}
\caption{General bit permutation.}
\Description{General bit permutation.}
\label{fig_general_pattern}
\end{figure}

As Section~\ref{sub_bitLevel} will explain, the vectorizer detects bit-permutation patterns by tracing the origin of each output bit and verifying that all bits originate from the same source vector, with no duplication.
Such permutations are compacted into single vector assignments expressed using concatenation and slicing.
This transformation preserves the original wiring semantics, but reduces symbolic complexity.
As Example~\ref{ex_intermodule_pattern} shows, the proposed vectorizer supports the identification of general permutations across module boundaries.

\begin{example}
\label{ex_intermodule_pattern}
Figure~\ref{fig_intermodule_pattern} illustrates intermodule vectorization.
This component of the vectorizer addresses cases in which vectorizable behavior is distributed across multiple module instances.
In these designs, structurally identical submodules, such as \texttt{my\_buf} in the example, are instantiated repeatedly, with each instance operating on a single bit of a vector.
\end{example}

\begin{figure}[ht]
\includegraphics[width=1.0\columnwidth]{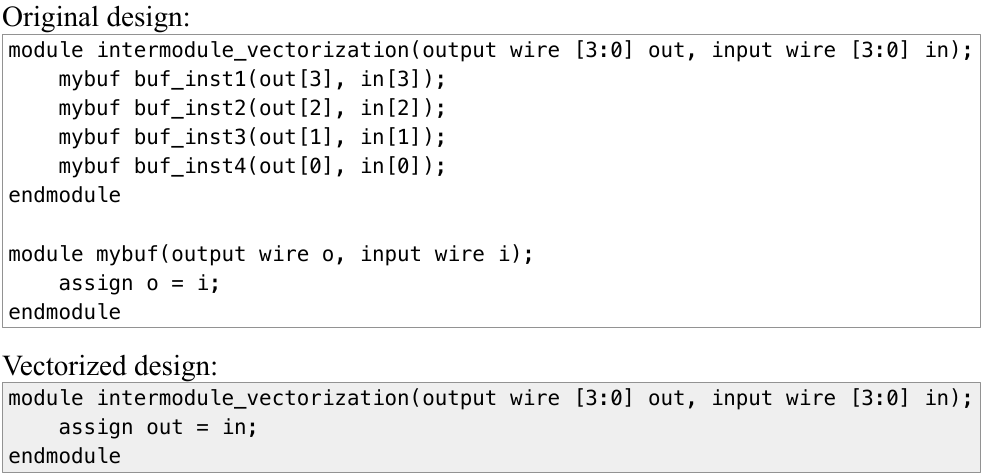}
\caption{Intermodule vectorization.}
\Description{Intermodule vectorization.}
\label{fig_intermodule_pattern}
\end{figure}

Section~\ref{sub_inlining} will show that the vectorizer applies selective inlining to expose redundant structures within module boundaries.
This transformation enables subsequent bit-level or structural analyses to recognize the underlying vector pattern.
Once inlined, the repeated scalar instances are replaced by a single vector-level operation, without altering module interfaces or the synthesized hardware behavior.
Example~\ref{ex_complex_pattern} shows that vectorization can be extended to assignments involving any number of logic or arithmetic operations.

\begin{example}
\label{ex_complex_pattern}
Figure~\ref{fig_complex_pattern} illustrates a case of complex-pattern vectorization.
This class captures replicated bitwise logic or arithmetic expressions that operate independently on each bit of one or more input vectors.
Typical examples include conditional selects, logical combinations, or arithmetic operations expressed as per-bit assignments.
\end{example}

\begin{figure}[ht]
\includegraphics[width=1.0\columnwidth]{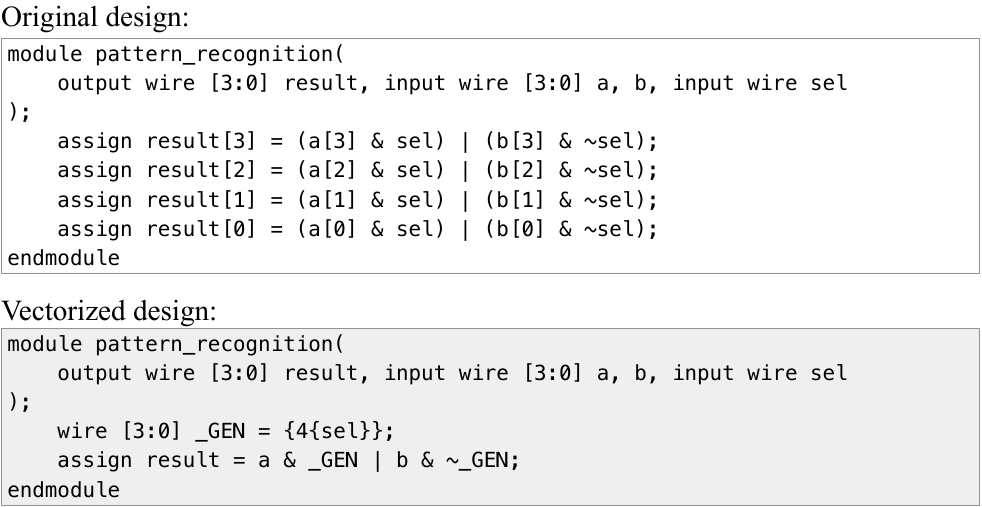}
\caption{Complex logic and arithmetic patterns.}
\Description{Complex logic and arithmetic patterns.}
\label{fig_complex_pattern}
\end{figure}

Section~\ref{sub_structural} explains how the complex patterns of Example~\ref{ex_complex_pattern} are recognized.
While these patterns cannot be identified as simple permutations, the vectorizer detects them by verifying that each output bit is computed by an independent and structurally equivalent logic cone.
When this condition holds, the repeated scalar expressions are merged into a single vector-level expression, yielding a more compact and symbolically efficient representation.

\paragraph{A Combined Pipeline.}
Examples~\ref{ex_general_pattern} through~\ref{ex_complex_pattern} highlight a key aspect of the proposed vectorization approach: it identifies vectorizable behavior in structurally recurrent code.
As discussed in Section~\ref{sec_intro}, this perspective contrasts with classic vectorization techniques, which primarily target patterns such as mappings and reductions expressed through loops.
This shift in perspective, from \emph{temporal} recurrences to \emph{spatial} recurrences, is necessary because standard Verilog lacks loop constructs with a clear execution semantics comparable to those of software languages.
Instead, word-level behavior is often expressed through replicated bit-level structures.
The proposed vectorizer is able to identify such recurrent structures in any combination of cases, including intermodule vectorization of permutations involving complex arithmetic and logical expressions.
Section~\ref{sec_impl} details how the patterns associated with each of these cases are detected and transformed into vectorized assignments.

\section{Implementation}
\label{sec_impl}

This paper's vectorizer is implemented on top of CIRCT and follows the flow seen in Figure~\ref{fig_overview}.
Notice that although the analyses and optimizations happen on the low-level representation of CIRCT, users still see them as a source-to-source transformation.
Thus, vectorization has the additional benefit of improving code readability.

\begin{figure}[ht]
\includegraphics[width=1.0\columnwidth]{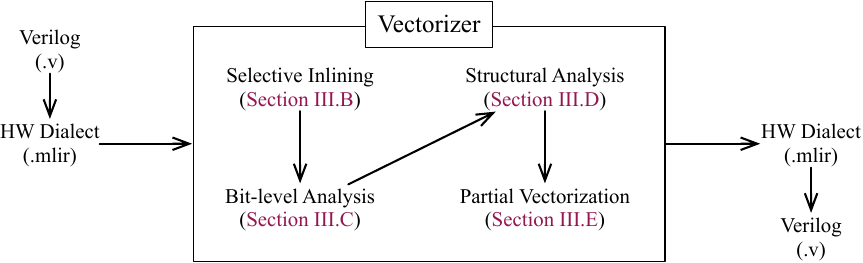}
\caption{The vectorization pipeline.}
\Description{The implementation of vectorization in the CIRCT Framework.}
\label{fig_overview}
\end{figure}

\subsection{CIRCT: The Underlying Infrastructure}
\label{sub_circt}

CIRCT is an open-source compiler infrastructure for hardware design, built on top of MLIR~\cite{Lattner21}.
It supports languages such as Verilog and Chisel.
CIRCT comprises multiple domain-specific intermediate representations, each implemented as an MLIR dialect.
This paper's implementation of vectorization takes place within the \emph{Hardware (\texttt{hw}) dialect}, which introduces data types such as fixed-width integers, arrays, and structs, and specifies module boundaries and ports.
In CIRCT, combinational and sequential behavior are expressed through specialized dialects such as \texttt{comb} and \texttt{seq}.
For example, \texttt{hw} defines modules (\texttt{hw.module}), instantiations (\texttt{hw.instance}), and wires (\texttt{hw.wire}), while \texttt{comb} provides logic operations such as \texttt{comb.and} and \texttt{comb.mux}.
Example~\ref{ex_circt} shows constructions in these different dialects.

\begin{example}
\label{ex_circt}
Figure~\ref{fig_exCirctAssignment}(a) shows a simple Verilog assignment, and Figure~\ref{fig_exCirctAssignment}(b) shows its equivalent specification in CIRCT.
When lowered to the \texttt{hw} and \texttt{comb} dialects, the design is represented structurally as four independent bit-level connections.
Each \texttt{comb.extract} operation isolates one bit from the input vector, and each bit is wired independently to the corresponding output position.  
The \texttt{hw} representation inherits some conventions from the LLVM ecosystem.
For instance, the program in Figure~\ref{fig_exCirctAssignment}(b) is in Static Single Assignment (SSA) form~\cite{Cytron91}, meaning that each variable is defined exactly once.
Moreover, every value is explicitly typed; for example, \texttt{i1} (a single bit) is the type of variable \texttt{\%out\_3}.
\end{example}

\begin{figure}[ht]
\includegraphics[width=1.0\columnwidth]{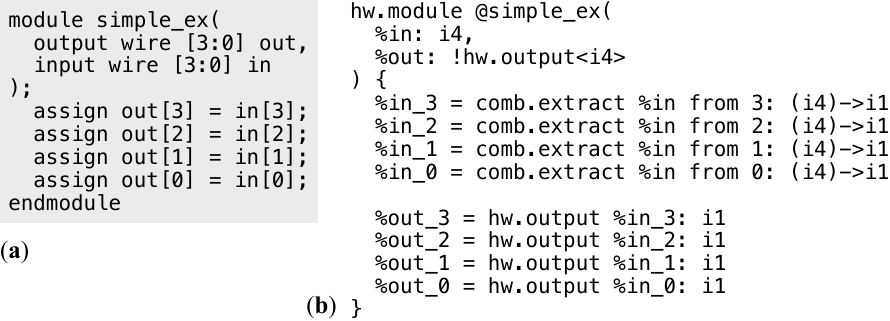}
\caption{(a) Verilog design. (b) Specification in CIRCT.}
\Description{(a) Verilog design. (b) Specification in CIRCT.}
\label{fig_exCirctAssignment}
\end{figure}

\subsection{Selective Inlining}
\label{sub_inlining}

The analyses and transformations described in Sections~\ref{sub_bitLevel}, \ref{sub_structural}, and~\ref{sub_vec} are \emph{intra-module}; that is, they operate within the boundaries of a single Verilog module.
To enable \emph{inter-module} vectorization, we employ \emph{selective inlining} of submodules.
Inlining consists of copying the body of a module into the point where it is instantiated.
By ``selective,'' we mean that a module may be inlined at some instantiation sites but not at others.
This approach is necessary because full inlining can cause code-size explosion.

Determining whether an inlining decision will contribute to vectorization is not possible without performing the inlining itself and rerunning the analyses from Sections~\ref{sub_bitLevel}, \ref{sub_structural}, and~\ref{sub_vec}.
To avoid such iterations while keeping compilation time low, inlining decisions are guided by two heuristics that estimate module regularity and recursive size:
\begin{itemize}
    \item The \textbf{Regularity} Analysis determines whether the callee module consists of simple combinational logic; that is, logic that the vectorizer is able to handle.
    \item The \textbf{Size} Analysis computes the total number of operations in the callee, including those in any instantiated submodules. Inlining is permitted if the cumulative operation count remains below 150 \texttt{hw} instructions.
\end{itemize}
The emphasis on regularity arises because the analyses described in Sections~\ref{sub_bitLevel} and~\ref{sub_structural} can only reason about modules composed of simple combinatorial operations.
The threshold of 150 instructions was chosen empirically, via an experiment that Section~\ref{sub_rq6} will describe.
This number is a compromise between vectorization opportunities and design size.

\subsection{Bit-Level Dataflow Analysis}
\label{sub_bitLevel}

This static analysis identifies groups of scalar assignments that collectively implement a permutation of bits, as illustrated in Figure~\ref{fig_general_pattern}.
The analysis propagates bit-level dataflow through chains of assignments and succeeds only if every output bit can be traced back to a single source vector, forming a one-to-one permutation.

Algorithm~\ref{alg_bit_permutation} summarizes the process of grouping multiple bit-level assignments into a single assignment.
The pass maintains a bit-to-source mapping $\pi$, inspired by \texttt{BitArray}, which associates each output bit with its originating value and source index.
If the mapping is not bijective, or if a bit passes through unsupported operations, the analysis fails.

\begin{algorithm}[htb]
\caption{Bit-Permutation Vectorization}
\label{alg_bit_permutation}
\begin{algorithmic}[1]

\STATE \textbf{Phase 1: Bit-Origin Analysis}
\FOR{$i = 0$ to $N-1$}
  \STATE Trace the dataflow of $out[i]$
  \IF{$out[i]$ does not originate from a single source vector}
    \STATE \textbf{return} \textsc{Fail}
  \ENDIF
  \STATE Record $\pi[i]$ such that $out[i] = in[\pi[i]]$
\ENDFOR

\IF{$\pi$ is not a bijection over $\{0,\dots,N-1\}$}
  \STATE \textbf{return} \textsc{Fail}
\ENDIF

\STATE \textbf{Phase 2: Greedy Grouping}
\STATE $i \gets 0$
\WHILE{$i < N$}
  \STATE $j \gets i$
  \WHILE{$j+1 < N$ \AND $\pi[j+1] = \pi[j] + 1$}
    \STATE $j \gets j + 1$
  \ENDWHILE
  \STATE Append segment $in[\pi[i]..\pi[j]]$
  \STATE $i \gets j + 1$
\ENDWHILE

\STATE \textbf{Phase 3: Emission}
\STATE Emit $out = \texttt{concat}(Segments)$
\STATE \textbf{return} \textsc{Success}
\end{algorithmic}
\end{algorithm}

Once Algorithm~\ref{alg_bit_permutation} identifies a valid permutation, the vectorizer constructs the mapping $\pi$ such that $out[i] = in[\pi[i]]$.
The transformation then applies a greedy grouping strategy that scans $\pi$ and groups maximal contiguous increasing subsequences into extract segments.
Each segment is emitted as a \texttt{comb.extract}, and the full result is assembled using a \texttt{comb.concat}.
Special cases are handled directly:
\begin{itemize}
\item For the identity permutation, the output vector is replaced by the input vector.
\item For a full reversal, the output is replaced by a single \texttt{comb.reverse}.
\item All other permutations are emitted as concatenations of extracted segments.
\end{itemize}
The pass itself does not attempt to merge extract ranges; this simplification is performed later by CIRCT's lowering and printing infrastructure.
Algorithm~\ref{alg_bit_permutation} does not trace through \texttt{comb.and} or \texttt{comb.mux}.
This tracing is performed by the structural analysis described in Section~\ref{sub_structural}.
Example~\ref{ex_general_vec_comb} explains how Algorithm~\ref{alg_bit_permutation} works.

\begin{example}
\label{ex_general_vec_comb}
Algorithm~\ref{alg_bit_permutation} operates on the \texttt{comb} MLIR dialect, which is part of CIRCT.
In our pipeline, Verilog is first lowered into \texttt{comb}. The resulting \texttt{comb} program is then vectorized and later emitted back into high-level Verilog.
Figure~\ref{fig_general_vec_comb} illustrates this process.
The bit-origin analysis of Algorithm~\ref{alg_bit_permutation} detects in the \texttt{comb} program the permutation corresponding to the assignment
\texttt{out = \{in[0], in[3], in[2], in[1]\}}.
During the greedy grouping phase, the algorithm identifies that bits \texttt{in[3:1]} form a contiguous slice that can be emitted as a single operation\footnote{The \texttt{comb.extract} instruction extracts a sequence of bits starting at a given position; for example,
\texttt{comb.extract \%in2 from 1 : (i4) -> i3} extracts three bits starting at position~1, corresponding to the slice \texttt{in[3:1]}.}.
Finally, the vectorized representation is emitted back into Verilog.
The resulting code has lower symbolic complexity, as three individual assignments are replaced with a single slice expression.
\end{example}

\begin{figure}[ht]
\includegraphics[width=1.0\columnwidth]{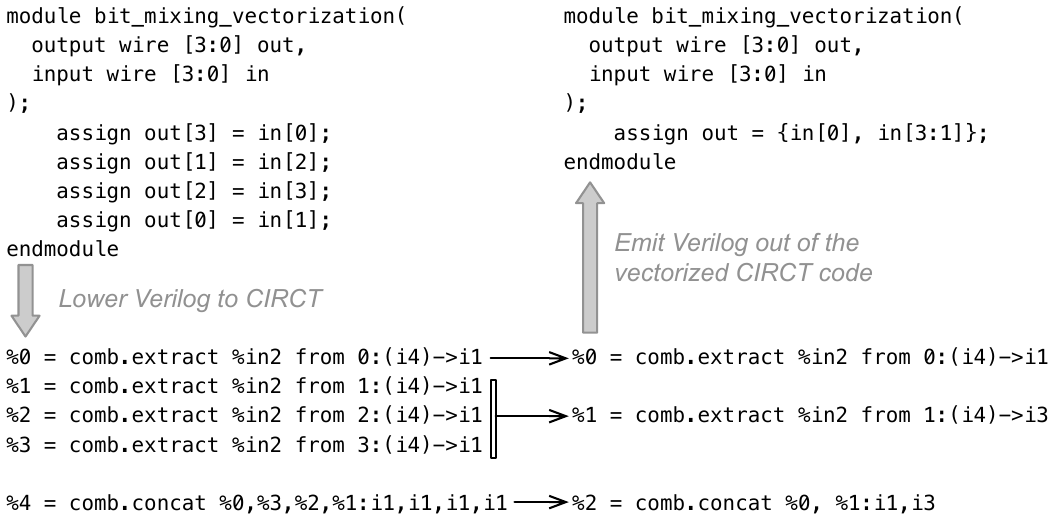}
\caption{Example of vectorizing a bit permutation. Individual assignments in Verilog are lowered to \texttt{comb} operations, grouped through bit-origin analysis, and emitted back as the slice-based expression \texttt{out = \{in[0], in[3:1]\}}.}
\Description{Vectorization of a bit permutation from Verilog to the comb dialect and back to Verilog using a slice expression.}
\label{fig_general_vec_comb}
\end{figure}

\begin{theorem}
\label{theo_bit_permutation}
Let $N$ be the bit-width of a target vector and $D$ be the maximum depth of the dataflow graph. Algorithm 1 (Bit-Permutation Vectorization) identifies and groups bit-level assignments into a word-level concatenation in $O(N \cdot D)$ time.
\end{theorem}

\begin{small}
\begin{quotation}
\textbf{Proof:} The algorithm operates in three phases. First, for each of the $N$ bits in the output vector, it performs a recursive dataflow trace to identify its origin. Since each trace is bounded by the graph depth $D$, this phase takes $O(N \cdot D)$. Second, it performs a linear scan ($O(N)$) to identify contiguous sequences of bits that form a valid permutation. Finally, it emits the resulting segments, which is at most $O(N)$ operations. The total complexity is dominated by the initial tracing, resulting in $O(N \cdot D)$.
\end{quotation}
\end{small}

\subsection{Structural Analysis}
\label{sub_structural}

The structural analysis identifies {\it logic cones}.
A logic cone is the set of upstream operations that influence a single scalar output bit.
Definition~\ref{def_logic_cone} formalizes this notion.

\begin{definition}[Logic Cone]
\label{def_logic_cone}
Let a circuit be represented as a directed acyclic graph whose nodes are
operations and whose edges represent bit-level data dependencies.
Let $b$ be a scalar output bit.
The \emph{logic cone} of $b$, denoted $LC(b)$, is defined recursively as follows:
\begin{description}
\item[Base case:] 
If $b$ is a primary input bit or a constant, then:
\[
LC(b) = \{b\}.
\]

\item[Recursive case:]
If $b$ is produced by an operation $op$ whose input bits are
$i_1, i_2, \dots, i_n$, then:
\[
LC(b) =
\{op\} \cup \bigcup_{k=1}^{n} LC(i_k).
\]
\end{description}
\end{definition}

Definition~\ref{def_logic_cone} implies that the logic cone of a bit $b$ contains the operation that produces $b$ together with $b$'s {\it backward slice}; that is, the logic cones of all bits that influence $b$'s computation.
Example~\ref{ex_logic_cone} illustrates this concept.

\begin{example}
\label{ex_logic_cone}
Figure~\ref{fig_logic_cone} shows the logic cone that defines bit \texttt{\%out}.
This example illustrates an important fact: logic cones are acyclic.
A logic cone is acyclic because it is defined over the
combinational portion of the circuit. Sequential elements, such as
registers, break feedback paths and therefore act as boundaries of
logic cones.
In the \texttt{comb} dialect of CIRCT, operations are guaranteed to be combinational, so the intermediate representation itself already enforces the absence of cycles.
In other words, logic cones in \texttt{comb} are always directed acyclic graphs by construction.
\end{example}

\begin{figure}[ht]
\includegraphics[width=1.0\columnwidth]{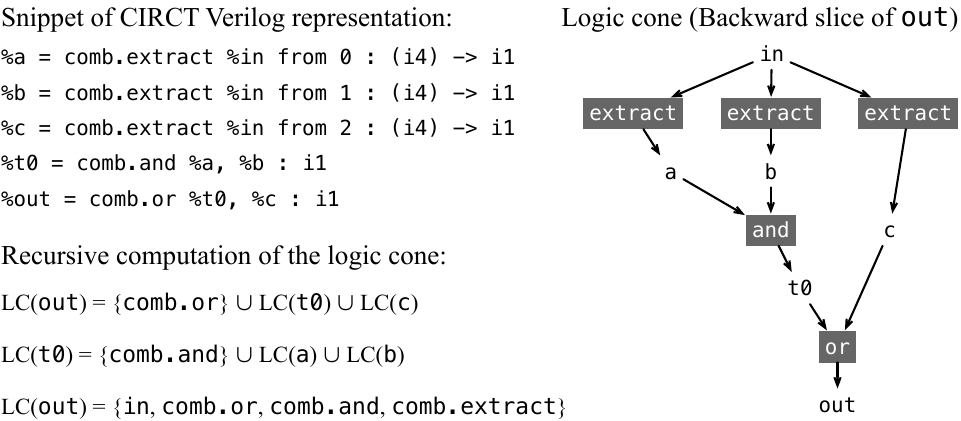}
\caption{Computation of logic cone.}
\Description{Computation of logic cone.}
\label{fig_logic_cone}
\end{figure}

\subsubsection{Grouping of Logic Cones.}
\label{sss_grouping}

Structural Analysis determines whether a set of bit-level operations can be safely collapsed into a single vector operation by verifying if:
\begin{enumerate}
\item The cones of all output bits are \emph{independent} (no cross-bit dependencies);
\item The cones are \emph{isomorphic} (identical structure with a uniform bit stride).
\end{enumerate}
When these two conditions hold, the vectorizer replaces $N$ replicated scalar operations with a single vector operation.
For a vector-typed output value with $N$ bits, the analysis constructs $N$ logic cones and performs the following checks:
\begin{itemize}
\item \textbf{Safety:} Cones must not overlap; shared intermediate values would indicate bit-to-bit dependencies (e.g., carry chains), which prevent vectorization.

\item \textbf{Equivalence:} Cones must exhibit the same operator structure and consistent index progression across the input vectors (e.g., $a[i]$, $b[i]$, $c[i]$).
\end{itemize}
If the cones are structurally isomorphic and semantically independent, then they are merged and lowered to the most concise vector operation available in the target dialect.

\begin{example}
\label{ex_structural}
Each output bit in Figure~\ref{fig_complex_pattern} is computed using the same logical structure applied independently to the corresponding bits of two input buses.
For a bit-vector output $out$, the \emph{logic cone} of $out[i]$ is the complete set of upstream operations that may affect its value.
Structural Analysis first ensures that cones are \emph{independent}, except for constants and invariant values (like inputs):
\[
\begin{aligned}
\forall i \neq j,\;
\mathit{Cone}(out[i]) \cap \mathit{Cone}(out[j]) = S
\;\wedge\; S \neq \emptyset
\Rightarrow \\
\forall a \in S,\; a \text{ is invariant}
\end{aligned}
\]
Next, it checks that all cones are \emph{isomorphic}:
\[
\mathit{Cone}(out[0]) \cong \mathit{Cone}(out[1]) \cong \dots \cong \mathit{Cone}(out[N-1]),
\]
meaning they contain identical operator structure and width progression.
This equivalence check performs a subgraph isomorphism test that verifies:
(i) the same operators appear in each cone,
(ii) with the same structural connectivity and operand ordering, and
(iii) bit extraction follows a consistent index stride (e.g., $a[i]$, $b[i]$, control signals).
In this case, all cones satisfy independence and isomorphism, enabling the replacement of $N$ replicated bitwise expressions with a single vector-level operation (e.g., a conditional select).
\end{example}

\subsubsection{Implementation Details.}
\label{sss_impl}

Algorithm~\ref{alg_cone} summarizes the implementation of structural vectorization.
The algorithm operates independently on each module $M$ that forms a Verilog design $P$ and targets vector-typed outputs. For a vector output $v$ of width $N$, the algorithm first constructs the \emph{backward cone} of each bit $v[i]$ (this process will be explained in Section~\ref{sss_backward}).

\begin{algorithm}[htb]
\caption{Structural Vectorization via Backward Cones}
\label{alg_cone}
\begin{algorithmic}[1]
\STATE \textbf{Procedure} \textsc{Vectorize}($Design\ P$)
    \FORALL{modules $M \in P$}
        \FORALL{outputs $v$ in $M$ with bitwidth $N > 1$}
            \STATE $C \gets [\;]$ \COMMENT{list of backward cones}
            \FOR{$i \gets 0$ \TO $N-1$}
                \STATE $C[i] \gets$ \textsc{BackwardCone}($v[i]$)
            \ENDFOR
            \IF{\textsc{IsIndependent}($C$) \textbf{and} \textsc{IsIsomorphic}($C$)}
                \STATE $Perm \gets$ \textsc{ExtractPermutation}($C$)
                \STATE $Expr \gets$ \textsc{BuildVectorExpr}($C[0], N$)
                \STATE \textsc{ReplaceWithVectorAssign}($v, Expr, Perm$)
            \ENDIF
        \ENDFOR
    \ENDFOR
\end{algorithmic}
\end{algorithm}

Once all cones are collected, Algorithm~\ref{alg_cone} checks two safety conditions (see Line 8). First, the cones must be \emph{independent}, meaning that no operation appears in more than one cone. This condition ensures that the bits do not share intermediate values (e.g., carry chains), which would prevent vectorization. Second, the cones must be \emph{isomorphic}, i.e., they must exhibit the same operator structure and a consistent progression of bit indices across their inputs. Intuitively, this guarantees that each bit performs the same computation over different positions of the input vectors.

If both conditions hold, Algorithm~\ref{alg_cone} extracts the index permutation that maps each output bit to its corresponding input position (Line 9). This step captures cases such as identity mappings, reversals, and general shuffles. The vector expression is then reconstructed from a representative cone (the least significant bit), lifting scalar operations to their vector equivalents (Line 10). Finally, the original scalar assignments are replaced with a single vector assignment that preserves the detected permutation (Line 11).


\begin{theorem}
\label{theo_structural}
Given a hardware module with $V$ operations and $E$ dependencies, the structural vectorization pass (Algorithms 2 and 3) correctly identifies isomorphic logic cones in $O(N \cdot (V + E))$ time.
\end{theorem}

\begin{small}
\begin{quotation}
\textbf{Proof:} The algorithm constructs $N$ separate backward logic cones, one for each bit of the output signal. Each backward traversal (Algorithm 3) explores the dependency graph, which in the worst case visits every vertex and edge, taking $O(V + E)$. For $N$ bits, this total construction time is $O(N \cdot (V + E))$. The subsequent isomorphism check (Algorithm 2) compares these $N$ cones. Since it only requires verifying identical operator types and linear index strides across the established cones, the comparison remains proportional to the size of the graph.
\end{quotation}
\end{small}

\subsubsection{Construction of Backward Cones}
\label{sss_backward}

Algorithm~\ref{alg_backward} shows the procedure \textsc{BackwardCone}, which constructs the dependency subgraph associated with a scalar value. Starting from a root bit, it performs a backward traversal of the dataflow graph by recursively visiting the defining operation of each value and its operands. The traversal continues until it reaches primary inputs, constants, or other values with no defining operation. The result is the set of operations that influence the computation of the root bit.

\begin{algorithm}[htb]
\caption{Backward Cone Construction}
\label{alg_backward}
\begin{algorithmic}[1]
\STATE \textbf{Procedure} \textsc{BackwardCone}($root$)
\STATE $Worklist \gets [root]$
\STATE $Cone \gets \emptyset$
\STATE $Visited \gets \emptyset$

\WHILE{$Worklist \neq \emptyset$}
    \STATE $v \gets$ pop($Worklist$)
    \IF{$v \in Visited$}
        \STATE \textbf{continue}
    \ENDIF
    \STATE add $v$ to $Visited$

    \STATE $op \gets$ defining operation of $v$
    \IF{$op \neq \varnothing$}
        \STATE add $op$ to $Cone$
        \FORALL{operands $u$ of $op$}
            \STATE push $u$ into $Worklist$
        \ENDFOR
    \ENDIF
\ENDWHILE
\STATE \textbf{return} $Cone$
\end{algorithmic}
\end{algorithm}

The algorithm uses a worklist to ensure that each value is visited at most once, thus avoiding redundant exploration in the presence of shared subexpressions. The resulting cone is a directed acyclic subgraph of the module's dataflow graph, as Figure~\ref{fig_logic_cone} shows. The cone captures all combinational dependencies required to compute the bit.
This representation is subsequently used by Algorithm~\ref{alg_cone} to compare cones across different bit positions. Independence is checked by verifying that cones do not share operations, while structural isomorphism is determined by comparing the shape of the cones and the relative indices of their inputs.

\begin{theorem}
\label{theo_partial}
The partial vectorization strategy (Algorithm 4) using a greedy windowing approach has a worst-case time complexity of $\Theta(N^3)$, where $N$ is the vector bit-width.
\end{theorem}

\begin{small}
\begin{quotation}
\textbf{Proof:} Algorithm 4 attempts to find the largest vectorizable sub-ranges when a full vectorization fails. It employs a nested loop structure: the outer loop iterates through starting bits $i \in [N-1, 0]$, and the inner loop searches for the largest suffix $j$. This creates $O(N^2)$ candidate sub-vectors. For each candidate, it invokes the vectorization predicates (Bit-Permutation or Structural), which themselves scale linearly $O(N)$ with respect to the width of the segment. Consequently, the total worst-case complexity is $O(N^2 \cdot N) = \Theta(N^3)$.
\end{quotation}
\end{small}

\subsection{Partial Vectorization}
\label{sub_vec}

If Bit-Level and Structural Analyses fail to vectorize the full output bus, the vectorizer attempts a best-effort approach.
Instead of requiring the entire vector to follow a recognized pattern, this strategy discovers \emph{contiguous sub-ranges} (\emph{chunks}) that can be vectorized independently.
Algorithm~\ref{alg_partial_vec} summarizes this final part of the vectorization pipeline.

\begin{algorithm}[htb]
\caption{Partial Vectorization (Greedy Segmentation)}
\label{alg_partial_vec}
\begin{algorithmic}[1]
\STATE \textbf{Procedure} \textsc{PartialVectorize}($OutputBus\ out, N$)
\STATE $i \gets N - 1$
\WHILE{$i \geq 0$}
    \STATE $found \gets \textbf{false}$
    \FOR{$j \gets 0$ \TO $i-1$}
        \STATE $Chunk \gets out[i:j]$
        \IF{\textsc{CheckBit}($Chunk$) \OR \textsc{CheckStructure}($Chunk$)}
            \STATE \textsc{ApplyVectorization}($Chunk$)
            \STATE $i \gets j - 1$
            \STATE $found \gets \textbf{true}$
            \STATE \textbf{break}
        \ENDIF
    \ENDFOR
    
    \IF{\textbf{not} $found$}
        \STATE \textsc{EmitScalar}($out[i]$)
        \STATE $i \gets i - 1$
    \ENDIF
\ENDWHILE
\end{algorithmic}
\end{algorithm}

The algorithm scans the output vector from the Most Significant Bit (MSB) to the Least Significant Bit (LSB) using a sliding window.
For each bit position $i$, it searches for the largest suffix $out[i:j]$ that satisfies either the Bit-Level or Structural vectorization criteria.
If such a chunk is found, it is vectorized immediately, and the scan resumes at bit $j-1$.
If no valid chunk is found, the window size is reduced, and the search continues.

\begin{example}
Figure~\ref{fig_mixedVectorization} shows an instance of partial vectorization.
The full four-bit permutation in lines 5-8 does not match a known pattern.
However, partial vectorization may still discover that:
$out[3:1] \gets(in[3:1])$ and that $out[0]$ is amenable to structural vectorization.
In this case, two-bit chunks are vectorized independently, and the final result is reconstructed using a unified assignment.
\end{example}

\begin{figure}[ht]
\includegraphics[width=1.0\columnwidth]{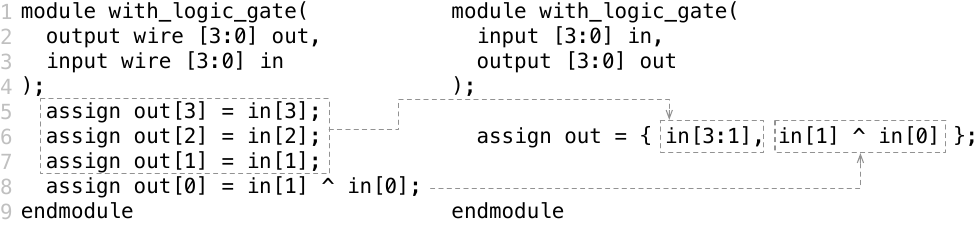}
\caption{Partial vectorization of two chunks of bits.}
\Description{Partial vectorization of two chunks of bits.}
\label{fig_mixedVectorization}
\end{figure}

\begin{theorem}
\label{theo_partial_complexity}
The greedy segmentation process implemented in Algorithm~\ref{alg_partial_vec} has a worst-case time complexity of $\mathcal{O}(N^2 \cdot T_{\text{check}}^{\max})$, which simplifies to $\Theta(N^3)$ when the vectorization predicate $T_{\text{check}}$ runs in linear time relative to the vector width $N$.
\end{theorem}

\begin{small}
\begin{quotation}
\textbf{Proof:} The algorithm identifies vectorizable sub-ranges by employing a nested loop structure. An outer loop iterates through starting bits $i \in [N-1, 0]$, while an inner loop evaluates potential end-bits $j < i$ to find the largest possible contiguous chunk. This results in $\mathcal{O}(N^2)$ candidate sub-vectors. For each candidate, the algorithm invokes a vectorization predicate (such as the Bit-Permutation or Structural checks). Given that these predicates scale linearly ($O(N)$) with the width of the segment being analyzed, the total worst-case complexity is $\mathcal{O}(N^2 \cdot N) = \Theta(N^3)$. As noted in Section~\ref{sub_rq3}, the practical performance often approaches $\mathcal{O}(N^2)$ due to the greedy nature of the discovery, which allows the algorithm to skip redundant checks once a maximal chunk is identified.
\end{quotation}
\end{small}

\section{Evaluation}
\label{sec_eval}

This section evaluates the following research questions:
\begin{itemize}
\item \textbf{RQ1:} What is the impact of vectorization on the instruction count of the \textsc{ChiBench} benchmarks?

\item \textbf{RQ2:} What are the absolute running time and asymptotic behavior of automatic vectorization?

\item \textbf{RQ3:} What is the impact of vectorization on Jasper's running time and memory consumption?

\item \textbf{RQ4:} What is the impact of vectorization on Genus's running time and memory consumption?

\item \textbf{RQ5:} What is the contribution of each technique seen in Section~\ref{sec_impl} on the number of vectorized patterns?

\item \textbf{RQ6:} How does the inlining threshold impact vectorization opportunities?
\end{itemize}
Below, we describe the experimental setup:

\noindent
\textbf{Hardware:}
All experiments were conducted on an AMD Ryzen 9 5900X processor with 32~GB of RAM, two L1 caches (384~KB each), one L2 cache (6~MB), and one L3 cache (64~MB), running Ubuntu Linux 22.04.5~LTS.

\noindent
\textbf{Software:}
Vectorization was implemented on top of CIRCT Release~192 (Oct-25).
Section~\ref{sub_rq3} uses the Jasper Formal Verification Platform release  2023.12p002, and Section~\ref{sub_rq4} uses the Genus Synthesis Solution version 21.19-s055\_1.

\noindent
\textbf{Benchmarks:}
Experiments use the \textsc{ChiBench} collection of Verilog designs~\cite{Sumitani24}.
\textsc{ChiBench} contains 49,599 designs mined from open-source repositories under permissible licenses.
Among these, 1,157 designs exhibited vectorization opportunities.
Unless stated otherwise, the numbers reported in this section refer to this subset of 1,157 Verilog files.

\noindent
\textbf{Correctness:}
Jasper Sequential Equivalence Checking (SEC) app has been used to demonstrate the semantic equivalence of each one of the 1,157 vectorized designs. 

\subsection{RQ1: Instruction Count}
\label{sub_rq1}

We report the size of Verilog designs in terms of the number of CIRCT's \texttt{hw} instructions that represent them.
Vectorization replaces multiple scalar instructions with one vectorized instruction.
Therefore, vectorization is expected to reduce the size of designs.
This subsection investigates this effect.

\noindent
\textbf{Discussion:}
Figure~\ref{fig_scatter_before_vs_after} compares the number of instructions before and after vectorization.
Among the 1,157 designs that presented vectorization opportunities, 904 experienced code-size reduction, 8 showed an increase in instruction count, and 245 remained unchanged, although they were transformed.

\begin{figure}[ht]
\includegraphics[width=0.9\columnwidth]{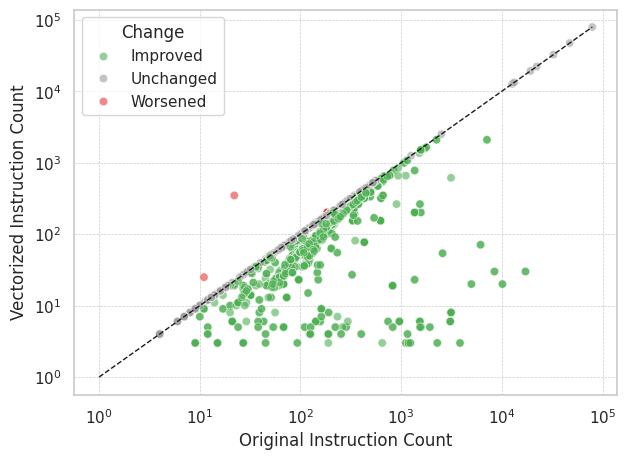}
\caption{Scatter plot comparing CIRCT instruction counts before and after vectorization. Points below the diagonal indicate reductions, points above indicate increases, and points on the diagonal indicate no change.}
\Description{Scatter plot comparing instruction counts before and after vectorization. Points below the diagonal indicate reductions, points above indicate increases, and points on the diagonal indicate no change.}
\label{fig_scatter_before_vs_after}
\end{figure}

Figure~\ref{fig_hist_reduction_pct_improved} shows the distribution of instruction-count reductions among the 904 designs that improved.
Most reductions range between 40\% and 50\%, but some designs achieved nearly 100\% reduction.
The four most dramatic improvements were observed in the following benchmarks:
\texttt{Chi\_11888}: 3,837 instructions reduced to 3 (99.92\%);
\texttt{Chi\_} \texttt{11336}: 2,283 reduced to 3 (99.87\%);
\texttt{Chi\_1203}: 17,149 reduced to 30 (99.82\%); and
\texttt{Chi\_11887}: 3,069 reduced to 6 (99.80\%).

\begin{figure}[ht]
\includegraphics[width=1\columnwidth]{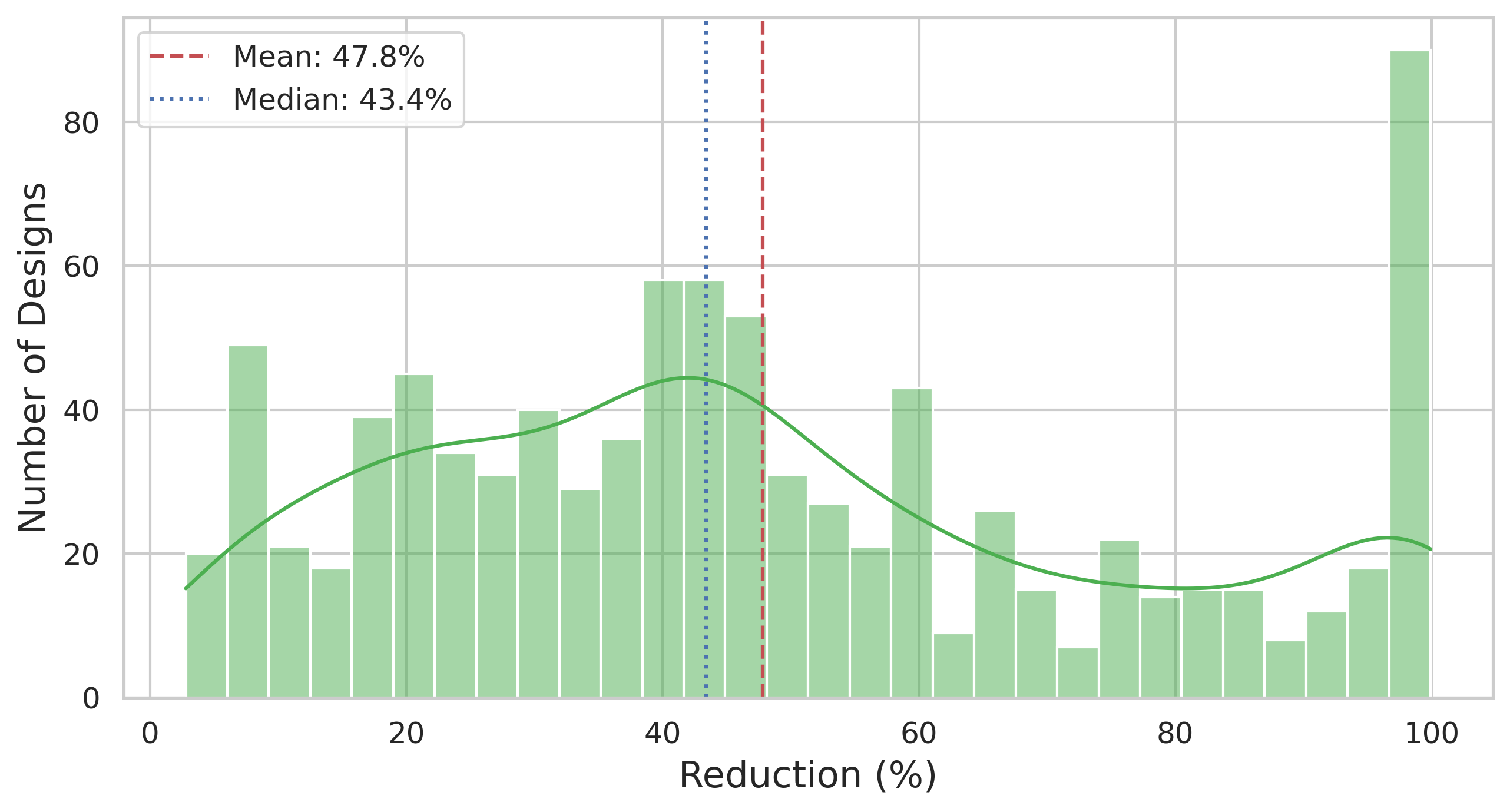}
\caption{Histogram of instruction-count reductions across designs that improved. Dashed and dotted lines indicate the mean (47.8\%) and median (43.4\%) reductions, respectively.}
\Description{Histogram of instruction-count reductions across designs that improved. Dashed and dotted lines indicate the mean (47.8\%) and median (43.4\%) reductions, respectively.}
\label{fig_hist_reduction_pct_improved}
\end{figure}

Not every application of vectorization leads to a smaller design, because maintaining semantic equivalence might require the insertion of intermediate wires.
We analyzed the eight cases where vectorization led to code-size increase and identified a single recurring design pattern. All instances of code-size expansion correspond to this same pattern.
These modules implement multiplexers that select 8-bit slices from a 64-bit input bus using a strided (non-contiguous) access pattern. In the original Verilog, these accesses are expressed via direct bit indexing (e.g., \texttt{data\_i[7]}, \texttt{data\_i[15]}), which maps to inexpensive wire references in the HW dialect. However, when vectorizing these designs, the vectorizer must explicitly materialize the input vector by inserting multiple ``\texttt{comb.extract}'' and ``\texttt{comb.concat}'' operations to gather the non-contiguous bits.
In this scenario, the overhead of preparing vector operands outweighs the benefits of grouping the multiplexer logic into vector operations, resulting in an increased instruction count.
A way to avoid this issue would be to incorporate a cost model into the vectorizer to account for data preparation overhead. For example, the vectorizer could avoid vectorization when operand accesses are non-contiguous and would require many extract/concat operations. However, since this pattern is rare, we chose not to add this complexity to the current implementation.

\subsection{RQ2: Compilation Time}
\label{sub_rq2}

This section investigates the running time of the vectorizer and its asymptotic behavior as a function of design size, to demonstrate that the proposed transformation is sufficiently practical to be used in production.

\noindent
\textbf{Discussion:}
Figure~\ref{fig_time_vs_instructions_loglog} shows the asymptotic behavior of the vectorization pass by plotting its runtime (in seconds) as a function of the number of vectorized \texttt{hw} instructions on a log-log scale. The results show a linear trend, with a coefficient of determination $R^2 = 0.982$, confirming that runtime grows proportionally to design size.
Vectorization process time is insignificant compared to the heavy computation steps, such as formal verification or synthesis.

\begin{figure}[ht]
\includegraphics[width=1\columnwidth]{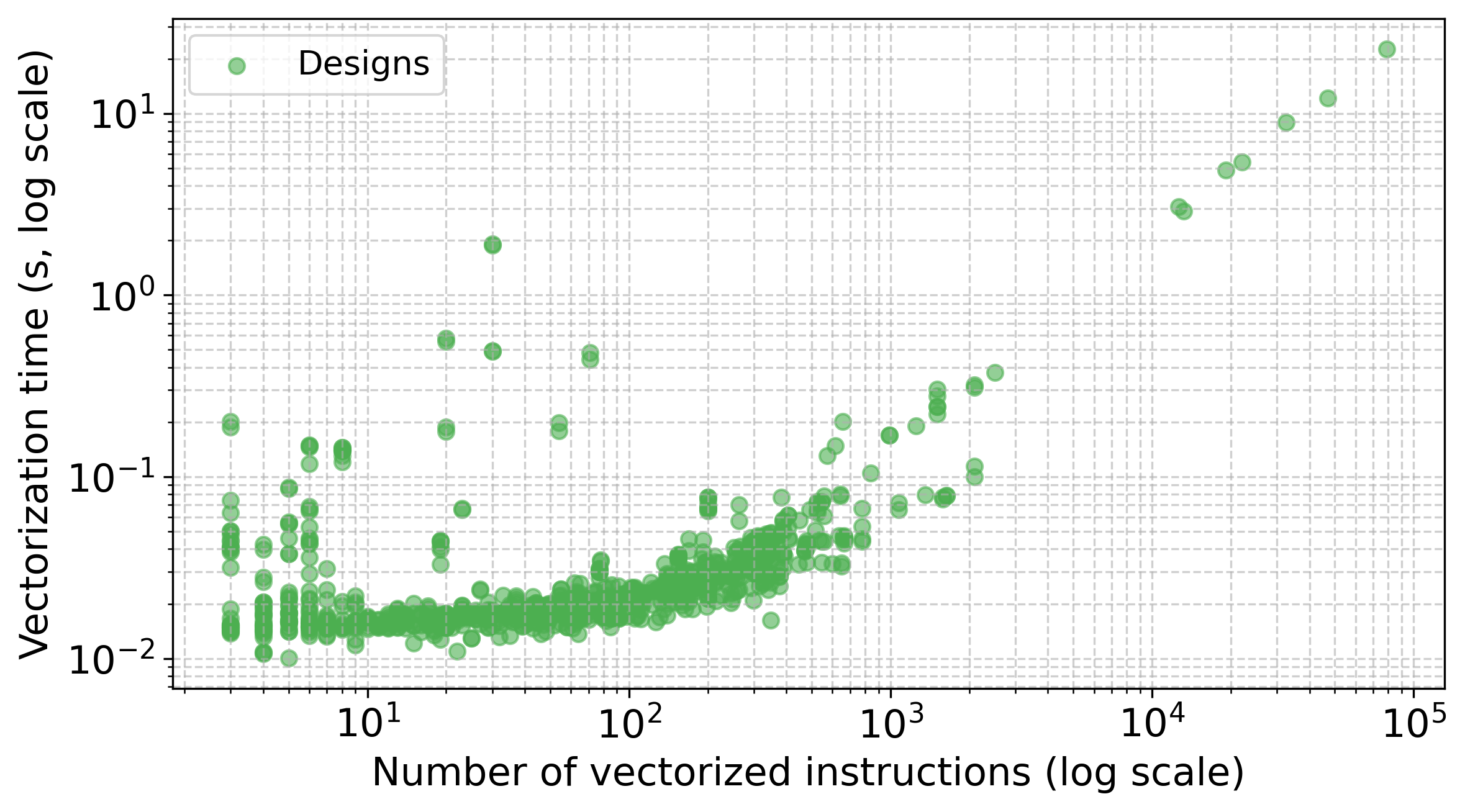}
\caption{Asymptotic behavior of the vectorization pass in a log-log scale.
The data points form an approximately straight line with slope $\approx 1$, indicating linear time complexity.}
\Description{Asymptotic behavior of the vectorization pass in a log-log scale.
The data points form an approximately straight line with slope $\approx 1$, indicating linear time complexity.}
\label{fig_time_vs_instructions_loglog}
\end{figure}

\subsection{RQ3: Impact on Jasper}
\label{sub_rq3}

Jasper is a formal verification platform that provides a suite of applications for C/C++ and RTL-level analysis, enabling designers to identify errors early in the development cycle. Within this framework, the RTL elaboration component is responsible for translating RTL descriptions into their corresponding netlist representations, consisting of nets, gates, and structural connectivity.
This section evaluates how vectorization speedups processes such as Jasper's elaboration.

\noindent
\textbf{Discussion:}
Figure~\ref{fig_jasper_mem_box} shows the impact of vectorization on Jasper's memory usage in the HDL elaboration step across the 1,157 benchmark designs. While the Original designs exhibit a wide distribution with heavy-tailed outliers, the vectorized versions are tightly clustered around a much lower median memory footprint. 

\begin{figure}[ht]
\includegraphics[width=1\columnwidth]{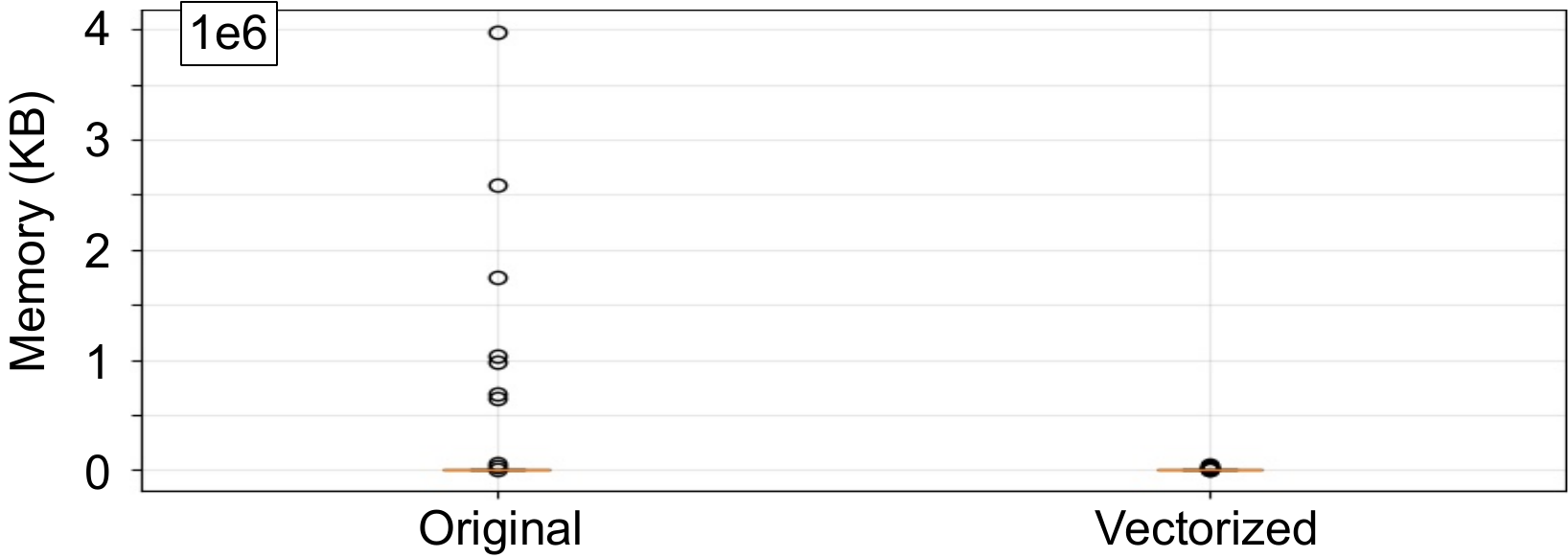}
\caption{Distribution of Jasper's elaboration memory (KB) for 1,157 original and vectorized designs.}
\Description{Distribution of Jasper's elaboration memory (KB) for 1,157 original and vectorized designs.}
\label{fig_jasper_mem_box}
\end{figure}

A similar pattern is observed for HDL elaboration runtime, as shown in Figure~\ref{fig_jasper_time_box}. It demonstrates that vectorization eliminates the high-runtime outliers present in the original designs. On average, vectorization reduced elaboration’s memory consumption by 51.30\% and runtime by 28.12\%.

\begin{figure}[ht]
\includegraphics[width=1\columnwidth]{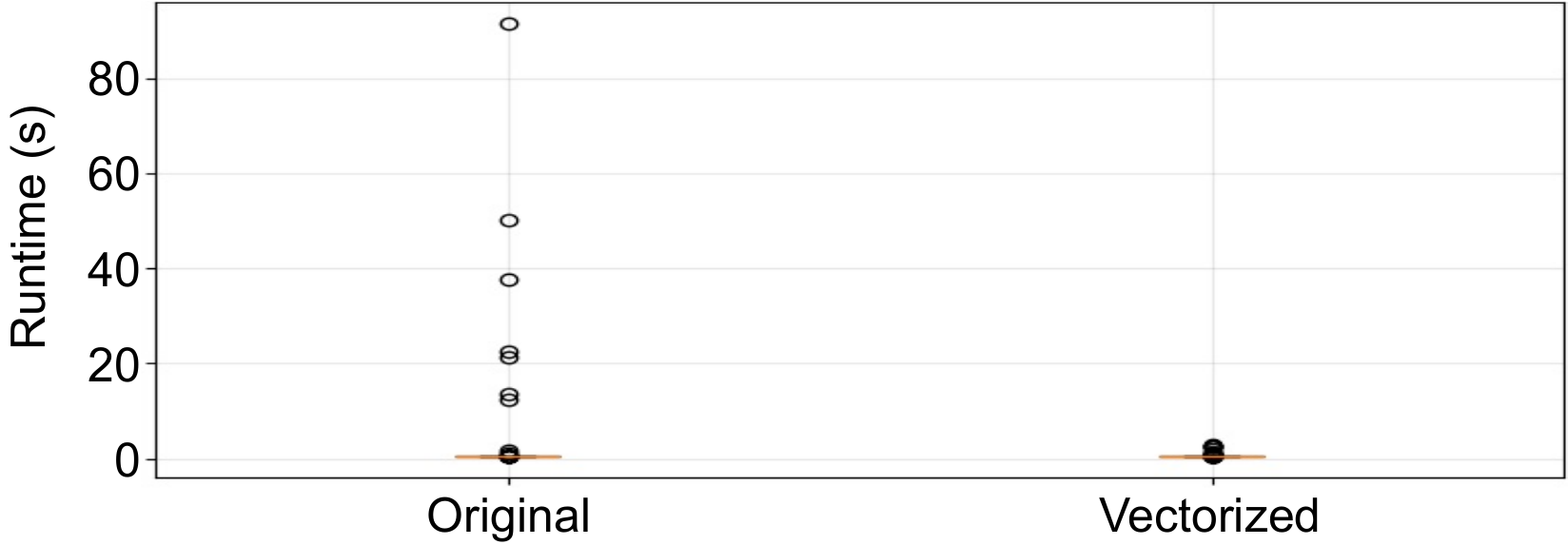}
\caption{Distribution of Jasper HDL elaboration runtime (s) for the 1,157 original and vectorized designs.}
\Description{Distribution of Jasper HDL elaboration runtime (s) for the 1,157 original and vectorized designs.}
\label{fig_jasper_time_box}
\end{figure}

\subsection{RQ4: Impact on Genus}
\label{sub_rq4}

The Cadence Genus Synthesis Solution is a physically aware RTL-to-gate synthesis tool that optimizes digital designs for power, performance, and area while ensuring timing closure. It processes a Verilog design through a two-stage pipeline.
In the first stage, elaboration, Genus parses the RTL code, resolves design hierarchy and constructs, and translates the high-level description into a gate-level netlist.
In the second state, synthesis, it performs physically aware optimizations, timing closure, and Design-for-Test (DFT) insertion to generate a gate-level netlist optimized for power, performance, and area.
Vectorization improves both phases. Memory is optimized because vectorized assignments are represented as single objects rather than as large sequences of bit-level operations. Reducing the number of assignments also accelerates processing. For example, when optimizing data paths, Genus can operate on entire vectors instead of individual bits. This section evaluates the impact of these improvements.

\noindent
\textbf{Discussion:}
After vectorization, we observe a geometric-mean improvement in elaboration time by 5.49\%.
Although general improvement is modest, individual benchmarks reveal large speedups. Figure~\ref{fig_improvements_genus} shows the 20 largest improvements in Genus memory consumption, elaboration time, and synthesis time. In one ChiBench design (from \url{https://github.com/awai54st/LUTNet}), memory consumption dropped by 75\%. In another design (from \url{https://github.com/ehw-fit/evoapproxlib}), synthesis time improved by 52\%. Improvements in elaboration time are more consistent across benchmarks: the 20 best improvements range from 45\% to 85\%, with a geometric mean of 64.5\%.

\begin{figure}[ht]
\includegraphics[width=1\columnwidth]{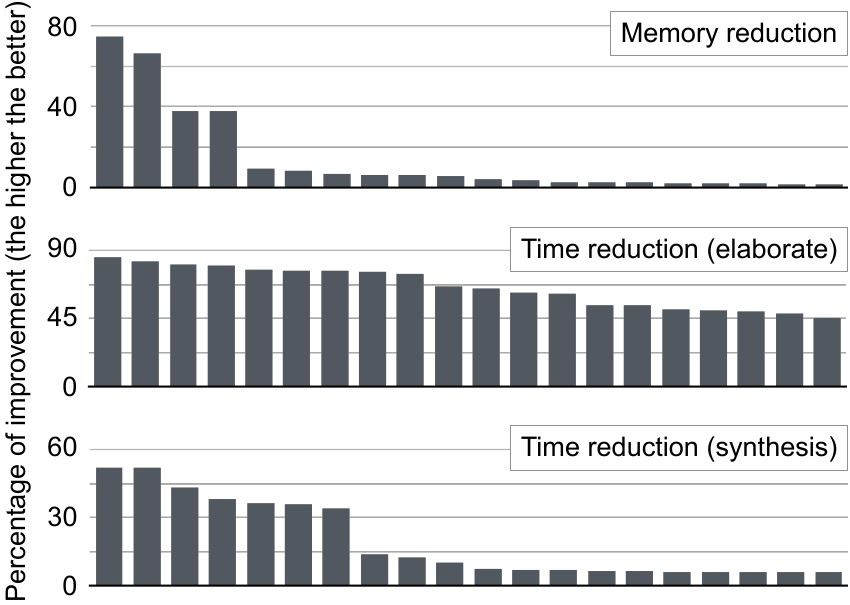}
\caption{The 20 largest observed improvements across three phases of the Genus pipeline. Each bar refers to a different benchmark. Benchmarks vary across the figures.}
\Description{The 20 largest observed improvements across three phases of the Genus pipeline. Each bar refers to a different benchmark. Benchmarks vary across the figures.}
\label{fig_improvements_genus}
\end{figure}

\subsection{RQ5: Impact of Individual Analyses}
\label{sub_rq5}

Each technique discussed in Section~\ref{sec_impl} contributes differently to the overall effectiveness of vectorization.
This section quantifies their individual impact.

\noindent
\textbf{Discussion.}
Figure~\ref{fig_separatedResults} breaks down the number of vectorized patterns according to the technique that enabled their recognition.
The majority of successful cases require a combination of bit-level and structural reasoning, following the partial vectorization strategy of Section~\ref{sub_vec}.
Selective inlining is also beneficial: while applicable to only 289 designs, it increased the number of vectorized patterns within these designs from 999 to 1,075.

\begin{figure}[ht]
\centering
\includegraphics[width=1\columnwidth]{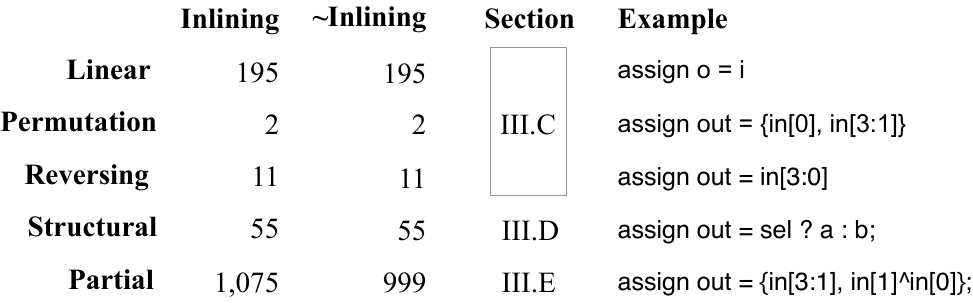}
\caption{Impact of the analyses in Section~\ref{sec_impl} on the number of vectorized patterns. Inlining was applicable in 289 out of 1,157 designs.}
\Description{Impact of the analyses in Section~\ref{sec_impl} on the number of vectorized patterns. Inlining was applicable in 289 out of 1,157 designs.}
\label{fig_separatedResults}
\end{figure}

Purely structural vectorization was comparatively rare, with 55 patterns across all 1,157 vectorized designs.
Patterns identified solely through the bit-level dataflow analysis were more common: we observed 208 such cases, 195 of which correspond to simple linear assignments where input and output indices align directly.
Overall, Figure~\ref{fig_separatedResults} highlights that effective vectorization typically requires multiple complementary analyses, as mixed-mode strategies account for the vast majority of recognized patterns.

\subsection{RQ6: Inlining and Vectorization}
\label{sub_rq6}

Section~\ref{sub_inlining} discussed selective inlining: a methodology that inlines modules to enable more vectorization opportunities.
That section mentioned that the current implementation of the vectorizer adopts an inlining threshold of 150 instructions; hence, only allowing code integration in as much as the size of the resuling module remains under 150 CIRCT instructions.
To settle for this number, we have evaluated the vectorizer using different inlining thresholds (30, 75, 150, 200, 300, and 400).
This section discusses this experiment.

\noindent
\textbf{Discussion.}
Figure~\ref{fig_analysis_dual_y} shows how the number of vectorization patterns varies with the inlining threshold. This count refers to individual vectorization opportunities, not to the number of files (a single design often contains multiple opportunities). Larger inlining thresholds lead to more vectorization opportunities; however, they also produce larger designs. The region of diminishing gains lies between 150 and 200 instructions. When increasing the threshold from 150 to 200 instructions, we observe eight new vectorization opportunities, but only one additional opportunity when moving the threshold from 200 to 300 instructions.

\begin{figure}[ht]
\centering
\includegraphics[width=1\columnwidth]{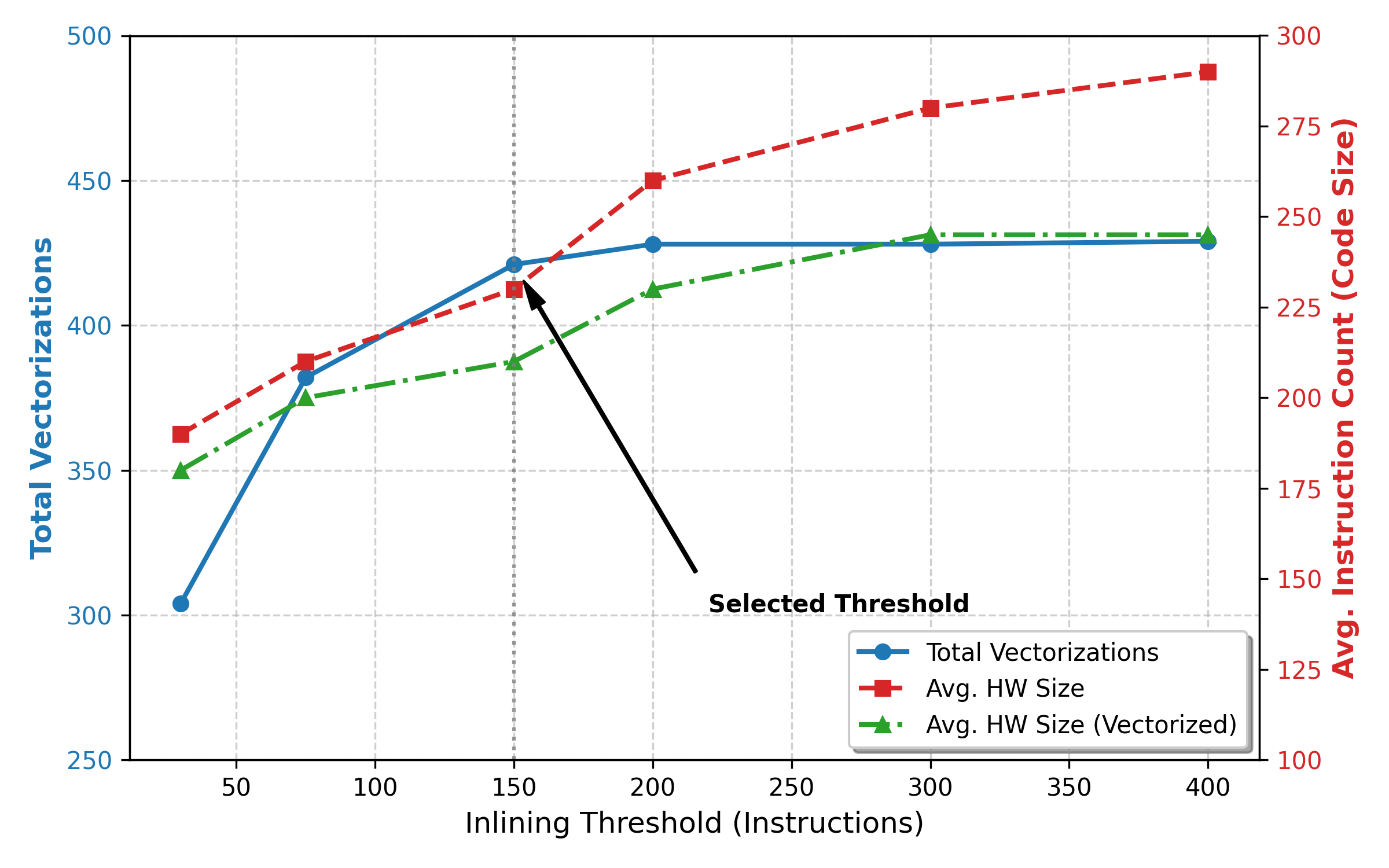}
\caption{Impact of the inlining threshold discussed in Section~\ref{sub_inlining} on the number of vectorization opportunities.}
\Description{Impact of the inlining threshold discussed in Section~\ref{sub_inlining} on the number of vectorization opportunities.}
\label{fig_analysis_dual_y}
\end{figure}

Notice that larger designs do not necessarily reduce the performance of EDA tools applied to them. As an example, although a lower threshold (e.g., 30) results in smaller absolute code sizes, it limits vectorization too much. Consequently, when we observe the performance of Jasper on the 31 designs that are vectorizable at both the 30 and 150 thresholds, we obtain better results using the latter. When the inlining threshold is 150 instructions, despite the code expansion caused by inlining, we notice that average memory consumption decreases (9,224.4 KB vs. 9,260.2 KB) and verification time improves (0.470 s vs. 0.474 s). Indeed, this threshold of 150 instructions yields the best overall performance for Jasper in our benchmarks.

\section{Related Work}
\label{sec_rw}

As we have discussed in Sections~\ref{sec_intro} and~\ref{sec_ovf}, the (spatial) vectorization of Verilog introduced in this paper is conceptually different from the (temporal) vectorization implemented by compilers such as \texttt{clang} and \texttt{gcc}.
Furthermore, to the best of our knowledge, no open-source RTL toolchain, including Yosys and CIRCT, performs vectorization as a semantics-preserving RTL-to-RTL transformation.
This type of transformation is also absent in HLS compilers such as LegUp, Bambu, and Xilinx Vivado.
Jasper, in contrast, includes vectorization in its most recent version (\texttt{2025.09FCS}), which we have tried. However, this feature does not run on top modules and currently does not include general bit permutation vectorization. Consequently, it did not show any effect in the benchmarks used in this paper. 
Nevertheless, our work connects to several research areas, including word-level reasoning for bit-vectors, recovery of word-level structure from bit-level/netlist representations, and algebraic rewriting techniques for structural circuit simplification.

\paragraph{Word-level reasoning and bit-vector solvers.}
Reasoning at the word (bit-vector) level can be far more efficient than bit-blasting in many verification tasks, leading to a long line of solvers and interpolation techniques that preserve word-level operators~\cite{Wang19,Griggio11}.
These works optimize the internal representations used by SMT and equivalence checking engines, whereas our method transforms Verilog source to \emph{expose} word-level structure before those engines operate.

\paragraph{Recovering word-level datapaths from bit-level nets.}
Several efforts address the inverse problem: given a low-level Boolean network (e.g., LUT or AIG netlist), reconstruct high-level words, carry chains, and datapath operators~\cite{Li13,Narayanan23}.
These ``word recovery'' and reverse-engineering techniques rely on pattern matching and structural heuristics to detect specific arithmetic or multiplexer patterns.
While conceptually related to the recognition component of our vectorizer, they are typically less general, operate after bit-blasting, and do not provide a semantics-preserving RTL-to-RTL rewrite.

\paragraph{Algebraic rewriting, AIG/FRAIG techniques, and structural compression.}
There is extensive work on algebraic rewriting and functional reduction using AIG-based representations~\cite{Kaufmann19,Yu18}.
Techniques such as cut rewriting, FRAIGs, and structural hashing reduce circuit size by optimizing boolean structure.
These approaches are complementary to ours: while AIG rewriting simplifies logic after lowering to bit-level form, our vectorizer restores and encodes word-level behavior early in the compilation flow, enabling more scalable reasoning by verification and synthesis tools.

\section{Conclusion}
\label{sec_conc}

This paper introduced a compiler transformation that performs vectorization of Verilog designs.
Although vectorization does not alter the synthesized hardware, it reduces the symbolic complexity of the underlying specification.
Our experimental results demonstrate that this reduction improves processes such as verification performance or synthesis optimization, as observed using the Jasper\textsuperscript{\textregistered} Formal Verification Platform and the Genus\textsuperscript{\texttrademark} Synthesis Solution.
Furthermore, because the technique operates as a source-to-source transformation, it also enhances the readability and maintainability of Verilog designs, providing practical benefits beyond tool performance.
Due to these benefits, the techniques proposed in this paper are being released as open-source extensions of the CIRCT framework.

\section*{Data Availability}
The vectorization algorithms discussed in this paper are publicly available as part of the CIRCT standard distribution.
To easy reproducibility, a standalone version of this implementation is publicly available at \url{https://github.com/lac-dcc/manticore/}, under the GPL 3.0 License.

\section*{Acknowledgment}
This project was supported by a grant from Cadence Design Systems.
The authors also acknowledge the support of FAPEMIG (Grant APQ-00440-23), CNPq (Grant \#444127/2024-0), and CAPES (\textsc{PrInt}).
We are grateful to the CIRCT community for incorporating this implementation into the main codebase.
We especially thank Hideto Ueno for carefully reviewing all the patches we submitted, and Fabian Schuiki, for helping us integrate the \texttt{comb.reverse} operation (necessary for implementing bit permutations) to CIRCT's \texttt{comb} dialect.

\bibliographystyle{plain}
\bibliography{references}

\end{document}